\documentclass[runningheads,12pt]{llncs}
\usepackage{graphicx}
\usepackage{tikz}
\usetikzlibrary{snakes}
\usepackage{xcolor}
\usepackage{hyperref}
\usepackage{xspace}
\usepackage{amssymb,amsmath,amsfonts,graphicx}
\usepackage{subcaption}

\newcommand{\vol}{{\rm vol}}

\usepackage{marginnote}
\usepackage{verbatim}

\usepackage{sectsty}
\allsectionsfont{\sffamily}

\begin{document}

\title{Iterated Global Models for Complex Networks\thanks{The first author acknowledges funding from an NSERC Discovery grant.}}
\titlerunning{Iterated Global Network Models}

\author{Anthony Bonato\inst{1} \and Erin Meger\inst{1}}

\authorrunning{A. Bonato, E. Meger}

\institute{Ryerson University, Toronto, ON, Canada \\
\email{\href{mailto:abonato@ryerson.ca}{abonato@ryerson.ca}, \href{mailto:erin.k.meger@ryerson.ca}{erin.k.meger@ryerson.ca}}}

\maketitle

\begin{abstract}
We introduce the Iterated Global model as a deterministic graph process that simulates several properties of complex networks. In this model, for every set $S$ of nodes of a prescribed cardinality, we add a new node that is adjacent to every node in $S$. We focus on the case where the size of $S$ is approximately half the number of nodes at each time-step, and we refer to this as the half-model. The half-model provably generate graphs that densify over time, have bad spectral expansion, and low diameter. We derive the clique, chromatic, and domination numbers of graphs generated by the model.
\end{abstract}

\section{Introduction}

Over the last two decades, research in modelling complex networks has become of great interest to mathematicians and theoretical computer scientists. Complex networks arise in technological, social, and biological contexts. The emergence of the study of complex networks such as the web graph and on-line social networks has focused attention on these large-scale graphs, and in the modeling and mining of their emergent properties; see \cite{bonato,bt,chung} for more on these models.

Two deterministic models of complex networks of particular interest to the current study were introduced: the \emph{Iterated Local Transitivity (ILT)} model and the \emph{Iterated Local Anti-Transitivity (ILAT)} model \cite{ilt,ilat}. Consider a social network where friendships have positive edge signs and adversarial relations have negative edge signs. A \emph{triad} is a set of three nodes in a signed network. A triad is said to be \emph{balanced} if the product of the edge signs is positive. Structural balance theory states that these networks seek to balance all triads \cite{ek}. The ILT and ILAT models were designed with balanced triads in mind. In the ILT model, nodes are \emph{cloned}, where nodes are adjacent to all neighbors of their parent node. In the ILAT model, nodes are \emph{anti-cloned}, where a new node is adjacent to all non-neighbors of it's parent node.  The ILT and ILAT models simulates many properties of social networks. For example, as shown in \cite{ilt}, graphs generated by the model densify over time (see \cite{les1} for more on densification), and exhibit bad spectral expansion (see \cite{estrada} for more on this topic in social networks). In addition, the ILT model generates graphs with the small-world property, which requires the graphs to have low diameter and dense neighbor sets. Both the ILT and ILAT models were unified in the recent context of Iterated Local Models in \cite{ilm}

The ILT, ILAT, and ILM models focused on considering the local structure of the graph and generating a new model iteratively from this structure. We now define a model that is independent of the structure of the initial graph but retains the iterative character of the previously defined models. We introduce the \emph{Iterated Global Models}, where a dominating node is added for each subset of nodes of a given cardinality.

Let $k\ge 1$ be an integer. The one parameter of the model is the initial, connected graph $G = G_0.$  At each time-step $t\ge 0,$ we create $G_{t+1}$ from $G_{t}$ in the following way: for each set of nodes of cardinality $\lfloor \frac{1}{k}|V(G_t)| \rfloor $, say $S$, add a new $v_S$ that is adjacent to each node of $S$. We name this process the \emph{$\frac{1}{k}$-model}. For ease of notation and for consistency with earlier chapters, we refer to newly added nodes in $G_{t+1}$ as \emph{clones}. Note that the clones form an independent set in $G_{t+1}.$

For the sake of clarity, we focus in this paper on the case $k=2,$ which we refer to as the \emph{half-model}. In the half-model, each new node is adjacent to approximately half of the existing network. See Figure~\ref{halfc4} for an example.
\begin{figure}[h!]
\begin{center}
\includegraphics[scale=1]{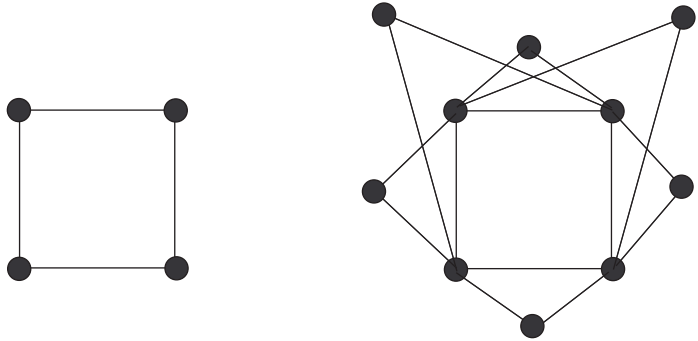}
\caption{One time-step of the half-model beginning with $C_4.$}\label{halfc4}
\end{center}
\end{figure}

While structural balance theory considers the importance of local ties, the half-model may be useful in analyzing complex networks where nodes interact via weaker, non-local ties. In social networks such as Twitter, Instagram, or Reddit, we may form a network of users where links are determined by likes, comments, or comments. For example, a user on Reddit may choose to comment on a fraction of the posts they read, which is reflective of the design of the half-model.

The paper is organized as follows. In Section~\ref{secc}, we prove that, as observed in complex networks, the half-model densifies over time and has bad spectral expansion. We also show that after five time-steps, graphs generated by the model have diameter at most 3. The half-model is of graph theoretic interest in its own right, and in Section~\ref{sech} we determine the clique, chromatic, and domination numbers of graphs generated by the model. We conclude with further directions to investigate for the half-model.

For a general reference on graph theory, the reader is directed to \cite{west}. For background on social and complex networks, see \cite{bonato,bt,chung1}. Throughout the paper, we consider finite, undirected graphs.

\section{Complex network properties of the half-model}\label{secc}

Our first result establishes the order and size of graphs generated by the half-model. We first recall Stirling's approximation for the factorial given by
$$n! \sim \sqrt{2\pi n} \left(\frac{n}{e}\right)^n.$$
Stirling's approximation may be used to derive an expression for the central binomial coefficient given by
$${{2n}\choose{n}} \sim \frac{2^{2n}}{\sqrt{\pi n}},$$
which may be derived directly and is part of folklore. Such an approximation will be useful in our analysis, and its usefulness has provided motivation for the study of the half-model as opposed to other values of $k$. For an exposition of the asymptotics of binomial coefficients, see the book \cite{spen}.

The number of nodes of $G_t$ is denoted $n_t$, the number of edges is denoted $e_{t}$.

\begin{theorem}\label{en1}
The order and size of the graph $G_t$ in the half-model are given by the following, respectively:
$$n_t \sim  \binom{n_{t-1}}{\left \lfloor \frac{n_{t-1}}{2} \right \rfloor} \hspace{1cm} \textrm{and} \hspace{1cm} e_t \sim  \binom{n_{t-1}}{\left \lfloor \frac{n_{t-1}}{2} \right \rfloor} \cdot \left \lfloor \frac{n_{t-1}}{2} \right \rfloor.$$
\end{theorem}
Before we give the proof of Theorem~\ref{en1}, we simplify notation by defining the function
$$ \alpha_t = \binom{n_t}{\left \lfloor \frac{n_t}{2} \right\rfloor } .  $$

\begin{proof}
We begin with the order of $G_t$. By the definition of the model, at each time-step $t\ge 1,$ we add one node for each set of size $\left \lfloor \frac{n_{t-1}}{2} \right \rfloor$. Hence, we derive the following sum given by
\[ n_t = n_0 + \sum_{i = 1}^{t} \alpha_{i-1}.\]
The term $\alpha_{t-1}$ will dominate the rest of the summation, which gives us the desired expression for the order of $G_t$.

Next, we determine the size of $G_t$. Each new node added is adjacent to a set of size $\left \lfloor \frac{n_{t-1}}{2} \right \rfloor$, and we add $\alpha_{t-1}$ nodes, so we obtain the following recursive formula for the number of edges at time-step $t:$

\[ e_t = e_{t-1} + \left \lfloor \frac{n_{t-1}}{2} \right \rfloor  \alpha_{t-1}.  \]
We observe that the second term dominates the sum, and the result follows. \qed
\end{proof}

We say that a network \emph{densifies} if the limit of the ratio of edges to nodes is unbounded. Densification power laws in complex networks were first reported in \cite{les1}. From Theorem~\ref{en1} we have the following result.

\begin{corollary}
The half-model densifies with time.
\end{corollary}

\begin{proof}
By Theorem~\ref{en1}, we have that
\[  \frac{e_t}{n_t} \sim \frac{\alpha_{t-1}\cdot \left \lfloor \frac{n_{t-1}}{2} \right \rfloor}{\alpha_{t-1}} = \left \lfloor\frac{n_{t-1}}{2} \right \rfloor ,\]
which tends to infinity with $t$. \qed \end{proof}

For a graph $G$ and sets of nodes $X,Y \subseteq V(G)$, define $E(X,Y)$ to be the set of edges in $G$ with one endpoint in $X$ and the other in $Y.$ For simplicity, we write $E(X)=E(X,X).$ Let $A$ denote the adjacency matrix and $D$ denote the diagonal degree matrix of a graph $G$. The \emph{normalized Laplacian} of $G$ is
\[ \mathcal{L} = I - D^{-1/2}AD^{-1/2}.\]
Let $0 = \lambda_0 \leq \lambda_1 \leq \cdots \leq \lambda_{n-1} \leq 2$ denote
the eigenvalues of $\mathcal{L}$. The \emph{spectral gap} of the normalized Laplacian is defined as
\[
\lambda = \max\{ |\lambda_1 - 1|, |\lambda_{n-1} - 1| \}.
\]

We will use the expander mixing lemma for the normalized Laplacian~\cite{chung}. For sets of nodes $X$ and $Y$, we use the notation $\vol(X) = \sum_{v \in X} \deg(v)$ for the volume of $X$, $\overline{X} = V \setminus X$ for the complement of $X$, and, $e(X,Y)$ for the number of edges with one end in each of $X$ and $Y.$ Note that $X \cap Y$ need not be empty, and in this case, the edges completely contained in $X\cap Y$ are counted twice. In particular, $e(X,X) = 2 |E(X)|$.

\begin{lemma}[Expander mixing lemma]\cite{chung}\label{mix}
If $G$ is a graph with spectral gap $\lambda$, then, for all sets $X \subseteq V(G),$
\[
\left| e(X,X) - \frac{(\vol(X))^{2}}{\vol(G)} \right| \leq \lambda \frac{\vol(X)\vol(\overline{X})}{\vol(G)}.
\]
\end{lemma}

A spectral gap bounded away from zero is an indication of bad expansion properties, which is characteristic for social networks, \cite{estrada}. The next theorem represents a drastic departure from the good expansion found in binomial random graphs, where $\lambda = o(1)$~\cite{chung}.

\begin{theorem}\label{halfgap}
Graphs generated by the half-model satisfy $\lambda_t \sim 1,$ where $\lambda_t$ is the spectral gap of $G_t.$
\end{theorem}

\begin{proof}
Let $X=V(G_t)\backslash V(G_{t-1})$ be the set of cloned nodes added to $G_{t-1}$ to form $G_t$. Since $X$ is an independent set, we note that $e(X,X)=0$. We derive that
\begin{align*}
\textrm{Vol}(G_t) &= 2e_t \sim \alpha_{t-1} \cdot n_{t-1},\\
\textrm{Vol}(X) &= \alpha_{t-1} \cdot \left \lfloor \frac{n_{t-1}}{2}\right \rfloor ,\\
\textrm{Vol}(\overline{X}) &\sim \alpha_{t-1} \cdot \left \lfloor \frac{n_{t-1}}{2} \right \rfloor .\\
\end{align*}

Hence, by Lemma~\ref{mix}, we have that
\begin{align*}
\lambda_t &\ge \frac{(\textrm{Vol}(X))^{2}}{\textrm{Vol}(G_t)} \cdot \frac{\textrm{Vol}(G_t)}{\textrm{Vol}(X)\textrm{Vol}(\overline{X})}\\
&= \frac{\textrm{Vol}(X)}{\textrm{Vol}(\overline{X})} \\
&\sim \frac{\alpha_{t-1} \cdot \left \lfloor \frac{n_{t-1}}{2} \right \rfloor }{\alpha_{t-1} \cdot \left \lfloor \frac{n_{t-1}}{2} \right \rfloor } \\
&= 1,
\end{align*}
and the result follows. \qed
\end{proof}

We observe that the half-model has a small (in fact, constant) diameter as required for the small-world property. We first prove some results about the connectivity for graphs generated by this model.

\begin{lemma}\label{lemcon}
For all $t \geq 0,$ if $G_t$ is connected and $n_t \geq 2$, then $G_{t+1}$ is connected.	
\end{lemma}

\begin{proof}
If $v$ is a clone in $G_{t+1}$, then since $n_t \geq 2$, we have that $v$ is adjacent to at least one node $u$ in $V(G_{t}) \backslash V(G_{t+1})$. Since $G_t$ is connected by hypothesis, there exists a path from $u$ to any other node of $G_t$, and hence, there is such a path from $v$ to any node of $G_t$. Since the node $v$ was an arbitrary clone, we have shown there exists a path between any two nodes in $G_{t+1}$. \qed \end{proof}

In the case where $n_0=1$, then $G_0$ is $K_1$. Note that $G_1$ is $\overline{K_2},$ and $G_2$ is the disjoint union of two edges. In particular, $G_1$ and $G_2$ are not connected. The subsequent lemma will provide insight into how many iterations a disconnected graph requires before becoming connected.

\begin{lemma}\label{connectedlemma}
For all $t \geq 0,$ if $G_t$ is a graph with $n_t \geq 4,$ then $G_{t+1}$ is connected.
\end{lemma}
\begin{proof}
We proceed by a proof by contraposition. Suppose then that $G_{t+1}$ is disconnected, and so there exists two nodes $u,v$ in $G_{t+1}$ such that there is no path between them.
	
\smallskip

\noindent \emph{Case 1:} $u,v$ are both in $V(G_t)$.

\smallskip

In this case, there is no set of size $\left \lfloor \frac{n_{t-1}}{2}\right \rfloor$ that contains both $u$ and $v$, since otherwise, a clone in $G_{t+1}$ would be adjacent to both $u,v$. At each time-step $t$, we add a clone for every subset of size $\left \lfloor \frac{n_{t}}{2}\right \rfloor$; hence, it must be the case that $\left \lfloor \frac{n_{t}}{2}\right \rfloor  < 2$ which implies $n_t \leq 3$. This satisfies the negation of the predicate, and we have proved the result in this case.
	
\smallskip

\noindent \emph{Case 2:} Exactly one of $u$ or $v$ is not in $V(G_t)$; without loss of generality, say $u \in V(G_{t+1})\backslash V(G_t)$.

\smallskip
	
As $u$ is a clone it has degree $\left \lfloor \frac{n_{t-1}}{2}\right \rfloor$, and so has a neighbor $x$ in $G_t$, whenever $n_t \geq2$. Thus, there is no path from $x$ to $v$ in $G_t$, and we apply Case 1 using these two nodes.
	
\smallskip

\noindent \emph{Case 3:} Both $u,v$ are in $V(G_{t+1})\backslash V(G_t)$.

\smallskip
	
Since there are at least two clones it must be the case that $\alpha_{t} \geq2$, and so $n_t \geq 2$. There then exists some neighbor $x$ of $u$ in $G_t$ and some neighbor $y$ of $v$ in $G_t$. We then have that there is no path from $x$ to $y$ in $G_t$ and we apply Case 1 to these two nodes. The proof follows. \qed \end{proof}

Our next result proves the 2-connectivity of graphs generated by the half-model.

\begin{lemma}\label{2con}
The graph $G_t$ is 2-connected whenever $t\geq 4$, regardless of the input graph $G_0$.
\end{lemma}

\begin{proof}
Using the recursive formula for the number of edges at time $t$ in the proof of Theorem~\ref{en1}, for any graph $G_0$, we have at least four nodes after two time-steps. Using Lemma~\ref{connectedlemma}, we require at least one additional time-step to ensure connectivity. Thus, regardless of the input graph $G_0,$ it is the case that $G_t$ is connected for $t\geq 3$. We now claim that whenever a graph $G_t$ is connected, $G_{t+1}$ will be 2-connected.

\smallskip

\noindent \emph{Claim: }If $G_t$ is connected and $n_t\geq 4,$ then $G_{t+1}$ is 2-connected.	

\smallskip

If $G_t$ is 2-connected, then we are done since every node in the set $V(G_{t+1})\backslash V(G_t)$ has at least one neighbor in $V(G_t)$, and we may use the same two paths between those neighbors to find 2-connectivity. Suppose $G_t$ is at most 1-connected and thus let $u$ be a cut-node of $G_t$. Consider two nodes in $G_t$, say $a,b$, that have a shortest path through $u$. In $G_{t+1}$, there is some clone $z$ that is adjacent to both $a,b$. Therefore, we have two paths from $a$ to $b$, and the proofs of the claim and theorem follow. \qed \end{proof}

Our main result on the diameter of half-model graphs is the following.

\begin{theorem}
Suppose that $G_0$ has order at least 4. In the half-model, the diameter of $G_t$ for $t\ge 5,$ is at most three.
\end{theorem}

\begin{proof}
We consider the distance between two non-adjacent nodes $x, y \in V(G_t)$ in three cases.

\smallskip

\noindent \emph{Case 1}: $x,y \in V(G_{t-1})$.

\smallskip

There exists some set $S \subseteq V(G_{t-1})$ of cardinality $\left \lfloor \frac{n_{t-1}}{2} \right \rfloor$ containing both $x$ and $y$. Thus, the dominating node for this set $S$, $v_S$ is adjacent to both $x$ and $y$ so their distance is 2.

\smallskip

\noindent \emph{Case 2}: $x \in V(G_{t-1})$ and $y \notin V(G_{t-1})$.

\smallskip

There exists a node $z \in  N_{G_t}(y)$. There is some set $S \subseteq V(G_{t-1})$ so that $x,z \in S$. The node $v_S$ that dominates $S$ in $G_t$ is adjacent to both $x$ and $z$, so we have the path $yzv_Sx$. Hence, the distance between $x$ and $y$ is at most 3. The symmetric case where $y \in V(G_{t-1})$ and $x \notin V(G_{t-1})$ is analogous.

\noindent \emph{Case 3}: $x,y \notin V(G_{t-1})$.

Since $x,y$ are new nodes in time-step $t$, there must be two sets $S_x, S_y \subseteq V(G_{t-1}),$ where $x$ dominates $S_x$ and $y$ dominates $S_y$. If $S_x \bigcap S_y \neq \emptyset,$ then there is some node of $G_{t-1}$ adjacent to both $x$ and $y$, so their distance is 2. Suppose now that $S_x \bigcap S_y = \emptyset$.  Since $|S_x| = |S_y| = \left \lfloor \frac{n_{t-1}}{2} \right \rfloor$, it may be the case that there exists a node $z \notin S_x \cup S_y$.

Suppose first that there is no such node $z$. There must be some edge with one endpoint in $S_x$ and the other in $S_y,$ since otherwise, the graph would be disconnected, which contradicts Lemma~\ref{2con}. We call these two endpoints $a$ and $b$. We then have a path $xaby$ and the distance between $x$ and $y$ is 3.

If there is such a node $z$, then since $G_{t}$ is 2-connected by Lemma~\ref{2con}, $z$ cannot be a cut-node. Therefore, there must be some edge with one endpoint in $S_x$ and the other in $S_y$ and the distance between $x$ and $y$ is 3.\qed \end{proof}

\section{Graph Parameters for the Half-Model}\label{sech}

In this section, we discuss classical graph parameters for the half-model. For further background on these parameters, the reader is directed to \cite{west}. We begin by considering the independence and clique number.

\begin{theorem}
The independence number of $G_t$ is $\alpha_{t-1}$ and for the clique number we have
$$ \chi (G_t) \geq \mathrm{min}\left( \left \lfloor \frac{n_{t-1}}{2} \right \rfloor + 1  ,  \omega(G_0)+t  \right).$$
\end{theorem}

\begin{proof}
At each time-step $t$, all the cloned nodes form an independent set. The set of new nodes has order $\alpha_{t-1} \geq n_{t-1}$, so this set must be the largest independent set in $G_t$. Therefore, we derive that $\alpha(G_t) = \alpha_{t-1}.$

We next consider the clique number of $G_{t}$. At each time-step $t$, we add a dominating node to subsets of cardinality $\left \lfloor \frac{n_{t-1}}{2} \right \rfloor  $ from $G_{t-1}$. If the largest clique $K$ at time-step $t-1$ is contained in one such subset, then we have increased the order of $K$ by 1. However, the maximum degree of new nodes is $\left \lfloor \frac{n_{t-1}}{2} \right \rfloor$. Hence, we cannot increase the size of the largest clique to be larger than $\left \lfloor \frac{n_{t-1}}{2} \right \rfloor +1.$ \qed\end{proof}

We next give the chromatic number of the half-model.

\begin{theorem}
For the half-model, we have that the chromatic number is given by
	
	\[\chi(G_t)= \mathrm{min} \left( \chi(G_0)+t, \left \lfloor \frac{n_{t-1}}{2} \right \rfloor  +1 \right). \]

\end{theorem}

\begin{proof}
Suppose that $G_t$ is properly colored. Consider a \emph{rainbow} subset of nodes; that is, a set of nodes that requires each distinct color in the graph. Let the cardinality of this set be $r \ge 1$. When $r\le \left \lfloor \frac{n_{t-1}}{2} \right \rfloor $, any new clone added that contains this set in its neighbors will need a new color. When $r> \left \lfloor \frac{n_{t-1}}{2} \right \rfloor $, any new clone that is added will have a neighbor set smaller than the cardinality of the colors, which implies there will always be an available color. \qed\end{proof}

We finish by proving a result on the domination number of graphs generated by the half-model.

\begin{theorem}
The domination number of $G_t$ is

\[\gamma(G_t) = \left \lceil \frac{n_{t-1}}{2} \right \rceil  + 1.\]
\end{theorem}

\begin{proof}

We will first establish the upper bound $$\gamma(G_t) \le \left \lceil \frac{n_{t-1}}{2} \right \rceil  + 1.$$

Consider a set $S$ of $\lfloor \frac{n_{t-1}}{2} \rfloor$ non-clone nodes in $G_{t-1}$. The node $x_S$ dominates $S.$ The complement $T$ of $S$ in $V(G_{t-1})$ has cardinality $\left \lceil \frac{n_{t-1}}{2} \right \rceil.$ Hence, $T \cup \{x_S \}$ is the desired dominating set.

For the lower bound, we must show that $\gamma(G_t) > \left \lceil \frac{n_{t-1}}{2} \right \rceil .$ For a contradiction, suppose that some set of $\left \lceil \frac{n_{t-1}}{2} \right \rceil$-many nodes, say $X$, dominates $G_t.$ Suppose first that $X$ consists of non-clones. Regardless of the choice of $X$, there will be some set of non-clones, call it $T$, of size $\left \lfloor \frac{n_{t-1}}{2} \right \rfloor$ such that $X \cap T = \emptyset$. Thus, $x_T$ is not dominated, which is a contradiction.

\begin{figure}[ht!]
\begin{center}
\includegraphics[width=5cm]{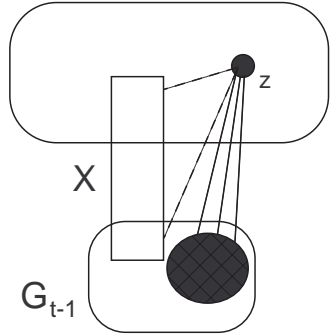}
\caption{The node $z$ is not adjacent to $X.$}\label{dom1}
\end{center}
\end{figure}
Suppose that $X$ contains at least one clone. There is a least one clone $z$ not adjacent to $X\cap V(G_{t-1})$, since $|X\cap V(G_{t-1})| < \left \lceil \frac{n_{t-1}}{2} \right \rceil$. See Figure~\ref{dom1}.
Any clone in $X$ is not adjacent to $z$, since the clones form an independent set. Therefore, $z$ is not dominated by $X$, which gives a contradiction. \qed\end{proof}

\section{Conclusion and further directions}

We introduced the Iterated Global Model (IGM) for complex networks. The IGM adds new nodes joined to $\lfloor \frac{1}{k}n_t \rfloor ,$ where $n_t$ is the number of nodes at time $t$. Our focus was the case $k=2,$ and we proved that graphs generated by the half-model exhibit densification, low distances, and bad spectral expansion as found in real-world, complex networks. We investigated various classical graph parameters for this model, including the clique, chromatic, and domination numbers.

Several open problems remain concerning properties of graphs generated by the half-model. Graph limits consider dense sequences of graphs and analyze their properties based on their homomorphism densities; see \cite{L}. Since the half-model generates dense sequences of graphs, it would be interesting to explore their graph limits. In the full version, we will consider the clustering coefficient of the half-model, analyze its subgraph counts, and degree distribution. Another interesting direction would be to generalize our results to integers $k > 2.$


\begin{thebibliography}{99}

\bibitem{bonato} A.\ Bonato, \emph{A Course on the Web Graph}, American Mathematical Society Graduate Studies Series in Mathematics, Providence, Rhode Island, 2008.

\bibitem{ilm}  A.\ Bonato, H.\ Chuangpishit, S.\ English, B.\ Kay, E.\ Meger, The iterated local model for social networks, Preprint 2020.

\bibitem{ilat} A.\ Bonato, E.\ Infeld, H.\ Pokhrel, P.\ Pra\l{}at, Common adversaries form alliances: modelling complex networks via anti-transitivity, In: \emph{Proceedings of WAW'17}, 2017.

\bibitem{ilt} A.\ Bonato, N.\ Hadi, P.\ Horn, P.\ Pra\l{}at, C.\ Wang, Models of on-line social networks, \emph{Internet Mathematics} \textbf{6} (2011) 285--313.

\bibitem {bt} A.\ Bonato, A.\ Tian, Complex networks and social networks, invited book chapter in: \emph{Social Networks}, editor E. Kranakis, Springer, Mathematics in Industry series, 2011.

\bibitem{chung} F.R.K.\ Chung, \emph{Spectral Graph Theory}, American Mathematical Society, Providence, Rhode Island, 1997.

\bibitem{chung1} F.R.K.\ Chung, L.\ Lu, \emph{Complex Graphs and Networks}, American Mathematical Society, Providence, Rhode Island, 2006.

\bibitem{ek} D.\ Easley, J.\ Kleinberg, \emph{Networks, Crowds, and Markets Reasoning about a Highly Connected World}, Cambridge University Press, 2010.

\bibitem{estrada} E.\ Estrada, Spectral scaling and good expansion properties in complex networks, \emph{Europhys. Lett.}
\textbf{73} (2006) 649--655.

\bibitem{les1} J.\ Leskovec, J.\ Kleinberg, C.\ Faloutsos, Graphs over time: densification Laws, shrinking diameters and possible explanations,
In: \emph{Proceedings of the 13th ACM SIGKDD International Conference on Knowledge Discovery and Data Mining}, 2005.

\bibitem{L} L.\ Lov\'asz, \emph{Large networks and graph limits}, American Mathematical Society,  Providence, RI, 2012.

\bibitem{spen} J.\ Spencer, L.\ Florescu, Asymptopia, American Mathematical Society, Providence, Rhode Island, 2014.

\bibitem{west} D.B.\ West, \emph{Introduction to Graph Theory, 2nd edition}, Prentice Hall, 2001.

\end{thebibliography}
\end{document}